%% file: main.tex
\documentclass[ aps,prl, reprint,superscriptaddress, floatfix,  
nofootinbib,
nobibnotes,
amsmath,amssymb,
]{revtex4-2}

\input{preamble.tex}

\begin{document}
\preprint{APS/123-QED}
\title{Practical Tomography of Multi-Time Processes}

\author{Abhinash Kumar Roy}
\email{abhinash.roy@students.mq.edu.au}

\affiliation{Department of Physical and Mathematical Sciences, Macquarie University, Sydney NSW, Australia.}
\author{Varun Srivastava}

\affiliation{Department of Physical and Mathematical Sciences, Macquarie University, Sydney NSW, Australia.}

\author{Christina Giarmatzi}

\affiliation{Department of Physical and Mathematical Sciences, Macquarie University, Sydney NSW, Australia.}


\author{Alexei Gilchrist}
\email{alexei.gilchrist@mq.edu.au}
\affiliation{Department of Physical and Mathematical Sciences, Macquarie University, Sydney NSW, Australia.}

\begin{abstract}

Characterising multi-time quantum processes is essential for analysing temporally correlated noise and for designing effective control and mitigation strategies. A complete operational description through multi-time process tomography requires an informationally complete set of probes, which necessarily includes non-deterministic intermediate operations. On present-day quantum devices, such operations are commonly implemented using mid-circuit measurements and reset, which are technologically limited and can introduce noise and overhead in terms of ancilla requirement. In this work, we study the minimal ancillary dimension required for complete characterisation of multi-time processes. We show that sequential interactions with a single  qubit ancilla can generate an informationally complete family of correlated probes for processes of arbitrary length, without requiring mid-circuit measurements or reset. Our result provides a resource-efficient route for complete multi-time process tomography and establishes that one qubit of coherent ancillary memory suffices for full reconstruction of arbitrary multi-time dynamics.
\end{abstract}

\maketitle

\emph{Introduction}: Temporally correlated noise is a major obstacle to the reliable operation of quantum devices \cite{Hashim_PRX_QCVV, modi_ibm, White2020,White2025whatcanunitary}. Unlike Markovian noise, non-Markovian dynamics retain memory of earlier system–environment interactions, so their effects cannot be captured by a sequence of independent channels \cite{Rivas2014,Breuer2016RMP,LI20181, costa2016,pollock_operational_markov,Pollock_pra,Milz2017,taranto2025higherorderquantumoperations}. Such memory effects invalidate the assumptions of standard error correction \cite{chuang00, Kam_2025}, gate characterisation \cite{srivastava2025blindspotsrandomizedbenchmarkingtemporal}, and control optimisation protocols \cite{White2020}, thereby posing a fundamental challenge to scalable quantum computing. A complete operational description of such dynamics is provided by the multi-time process matrix formalism, which generalizes quantum states and channels to sequences of interventions at different times and, in principle, enables the reconstruction of all experimentally accessible temporal correlations \cite{chiribella09b, oreshkov12, pollock_operational_markov, taranto2025higherorderquantumoperations, White_PRX_2025, modi_ibm, costa2016, PhysRevX.10.041049}. 

Access to the full process matrix is important not only for diagnosing memory effects, but also for developing memory-aware control and verification strategies for noisy quantum hardware \cite{srivastava2025blindspotsrandomizedbenchmarkingtemporal, White2020, White_PRX_2025, modi_ibm, Giarmatzi2021witnessingquantum, Giarmatzi2025multitimequantum, Goswami_2025}. However, complete multi-time process tomography remains experimentally demanding for two main reasons. First, the number of required experimental configurations grows rapidly with the number of probing times \cite{Milz_PRA_2018,White2022,santos2025driftresilientmidcircuitmeasurementstate}. Second, informational completeness requires non-deterministic intervention operations, such as measurement-and-repreparation maps, which on current platforms are typically implemented using mid-circuit measurement, feed-forward, and reset \cite{Giarmatzi2025multitimequantum}. Although such capabilities have enabled recent demonstrations of full process tomography, they remain slow, noisy, and technologically restrictive \cite{koh2025readouterrormitigationmidcircuit,Kenneth_Mid_Circuit}. Consequently, most present-day experiments access only restricted process matrices constructed from deterministic operations alone \cite{White2020,White2022,White2025whatcanunitary}, leaving a practical route to complete multi-time characterization on near-term hardware still lacking .

In this work, we study the minimal ancillary resources required for the complete characterisation of a fixed-size multi-time process, specified by the system dimension and the number of intervention times. We demonstrate that a single coherent qubit ancilla is sufficient to realise an informationally complete set of probes for arbitrary multi-time processes, without requiring either mid-circuit measurements or ancilla resets. We prove that sequential interactions between the system and a single qubit ancilla generate a family of correlated operations, constrained to have a matrix-product structure with fixed bond dimension, whose linear span nevertheless coincides with the full operator space needed for complete process reconstruction. This construction has implications beyond tomography. Since the accessible probes span the full multi-time operator space, arbitrary linear functionals of the process, including memory witnesses and other high-rank observables, can be evaluated from experimentally friendly probe statistics. The resulting protocol therefore offers a resource-efficient framework not only for complete multi-time tomography, but more broadly for the characterisation, verification, and control of temporally correlated quantum dynamics.

\emph{Multi-time processes}:
Consider a quantum system that evolves in time, possibly interacting with an environment. To probe this evolution one may intervene at different times by performing operations on the system. The process matrix formalism~\cite{oreshkov12,oreshkov15} provides a general operational framework to describe such multi-time experiments. In this formalism a sequence of laboratories (or \emph{labs}) $A_1,A_2,\ldots,A_N$ are placed along the system's time evolution (see Fig.~\ref{Fig:process_W}). Each lab implements a quantum instrument on the system associated with a linear map from an input Hilbert space $\mathcal{H}_i^{I}$ to an output Hilbert space $\mathcal{H}_i^{O}$, and  with a classical outcome $a_i$. The instrument implemented in lab $A_i$ is described by a completely positive map $\mathcal{M}_{a_i}:\mathcal{H}_i^{I}\rightarrow\mathcal{H}_i^{O}$, represented by its Choi operator \cite{choi_completely_1975,jamio72},
\begin{align}
M_{a_i}
=
(\mathcal{I}\otimes \mathcal{M}_{a_i})
\bigl(
|\mathbb{I}\rangle\langle \mathbb{I}|
\bigr),
\end{align}
which belongs to $\mathcal{L}(\mathcal{H}_i^{I}\otimes \mathcal{H}_i^{O})$. Here \(\mathcal{I}\) denotes the identity map on \(\mathcal{L}(\mathcal{H}_i^{I})\), and
$|\mathbb{I}\rangle=\sum_{n}|n\rangle\otimes |n\rangle$
is the unnormalised maximally entangled state on \(\mathcal{H}_i^{I}\otimes \mathcal{H}_i^{I}\), (equivalently the vectorised identity operator).

The multi-time process is fully characterised by a positive operator $W$, called the \emph{process matrix}, which belongs to  $\mathcal{L}\left(\bigotimes_{i=1}^{N}(\mathcal{H}_i^{I}\otimes\mathcal{H}_i^{O})\right)$. Given a set of operations $\{M_{a_i}\}$ implemented at each lab, the joint probability of observing outcomes $a_1,\ldots,a_N$ is given by the generalized Born rule \cite{Shrapnel_2018}

\begin{equation}
\label{born rule}
p(a_1,\ldots,a_N)
=
\mathrm{Tr}\!\left[
W^{T}
\left(
M_{a_1}\otimes \cdots \otimes M_{a_N}
\right)
\right].
\end{equation}
 A process matrix generalises quantum states and dynamical maps to the multi-time scenario and is consistent with classical stochastic processes in the appropriate limit~\cite{chiribella09b,pollock_operational_markov,Pollock_pra,costa2016,PhysRevX.10.041049}. If the operations implemented across the labs form an informationally complete set, the process matrix $W$ can be reconstructed from the observed probabilities. The reconstruction of $W$ from experimental data collected in the labs is known as multi-time process tomography, which we briefly review below.

\begin{figure}[t]
    \centering
    \includegraphics[width=0.4\textwidth]{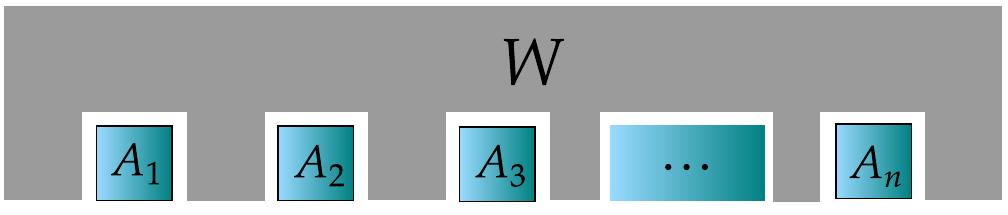}
    \caption{Schematic representing a multi-time process $W$ where the transformations might in general be correlated in time. The probes are depicted by $A_i$ and have associated input and output Hilbert spaces $\mathcal{H}_{A_i^{I}}$ and $\mathcal{H}_{A_i^{O}}$ respectively, where an arbitrary operation can be performed. Multi-time process tomography refers to reconstructing an unknown $W$ through implementing informationally complete set of operations at the probes.}
    \label{Fig:process_W}
\end{figure}

\emph{Multi-time process tomography}: In order to characterise a multi-time process one must implement an informationally complete set of operations across the laboratories. Consider $T_a$ to denote the Choi of a \emph{set} of experimental probe operations and outcomes at each lab, where $a$ is an index over the outcomes in the set. $T_a$ is a positive semi-definite operator in the space $\mathcal{L}\left(\bigotimes_{i=1}^{N}(\mathcal{H}_i^{I}\otimes\mathcal{H}_i^{O})\right)$. For example, in Eq.~\eqref{born rule}, $T_a$ is the tensor product of Choi operators from the set of operations $\{M_{a_i}\}$ with $a \equiv \{a_1,\cdots,a_N\}$. In general, the probes could represent correlated operations between labs known as testers or superinstruments \cite{taranto2025higherorderquantumoperations} (Also see Supplementary Material \cite{supplementary}). If this set is informationally complete on the full space, the corresponding probabilities $p_a = \mathrm{Tr}[W^{T} T_a]$
contain sufficient information to reconstruct the unknown process matrix $W$. In linear inversion tomography the reconstruction is obtained by expanding the vectorised process matrix in a dual frame of the probe operators,
\begin{equation}
|W\rangle = \sum_a p_a\,|D_a\rangle ,
\end{equation}
where $|D_a\rangle$ are the corresponding dual vectors satisfying $\sum_a |D_a\rangle\!\langle T_a| = I$, and $|T_a\rangle$ denotes the vectorised form of the operator $T_a$, defined as $I\otimes T_a \ket{\mathbb{I}}$. A canonical choice of dual vectors can be constructed using the frame operator $F = \sum_a |T_a\rangle\!\langle T_a|$, for which $|D_a\rangle = F^{-1}|T_a\rangle$. The probe set $\{T_a\}$ is informationally complete if and only if the frame operator $F$ is invertible, which is equivalent to the set of operations $\{T_a\}$ spanning the space $\mathcal{L}\left(\bigotimes_{i=1}^{N}(\mathcal{H}_i^{I}\otimes\mathcal{H}_i^{O})\right)$, ensuring that the process matrix can be uniquely reconstructed from the measured probabilities.

\emph{Informationally complete requirement and restricted operations}:
Consider first an operation in a single lab acting on a $d$-dimensional system, described by a linear, completely positive map $\mathcal{M}_a:\mathcal{H}^I\to\mathcal{H}^O$, where $a$ labels the classical outcome (for deterministic operations, this label simply specifies the chosen setting). In the Choi representation, this operation is associated with a positive semi-definite operator $M_a^{IO}\in\mathcal{L}(\mathcal{H}^I\otimes\mathcal{H}^O)$. Since the real vector space of Hermitian operators on $\mathcal{H}^I\otimes\mathcal{H}^O$ has dimension $d^4$, an informationally complete set of lab operations must span a $d^4$-dimensional space. Equivalently, one requires $d^4$ linearly independent Choi operators. This can be observed, for instance, by expanding operators in a basis such as $\{\sigma_\alpha\otimes\sigma_\beta\}$, where $\{\sigma_\alpha\}$ denotes a generalised Pauli operator basis on a $d$-dimensional Hilbert space, where $\sigma_0=I$ and $\forall i \neq 0$, $\sigma_i$ is traceless. Therefore, any informationally complete set of operations for a single lab must contain at least $d^4$ linearly independent elements.


If we restrict the intermediate probes to unitary operations acting on the system then the accessible set of lab operations is $\mathcal{U}=\{\,|U\rangle\langle U|:U^\dagger U = UU^\dagger = \mathbb{I}\,\}$. In the Choi representation, the unitarity constraint is equivalent to requiring that tracing over either the input or the output space yields the identity operator. In an expansion in a generalised Pauli operator basis, this implies that all coefficients multiplying terms of the form $\mathbb{I}\otimes \sigma_i$ and $\sigma_i\otimes \mathbb{I}$ must vanish, where $i \neq 0$. As a result, the space spanned by local unitary operations has dimension only $(d^2-1)^2+1$ (see the Supplementary Material~\cite{supplementary} for a proof). Hence, unitary operations do not form an informationally complete set of operations and are therefore insufficient for complete reconstruction of a multi-time process. For instance, unitary operations would not generate outcomes corresponding to entropy decreasing operations like amplitude damping channels. Nevertheless,  unitary probes are still operationally meaningful. In particular, they define a restricted process describing the statistics accessible when the intermediate interventions are limited to unitary (or mixed unitary) operations, followed by a final measurement.  Such restricted processes have been recently investigated for witnessing temporal correlations and for bounding their strength \cite{Milz_PRA_2018,White2025whatcanunitary}. It is also worth noting that the dimension of the space spanned by local unitaries has the same leading-order scaling in $d$ as that of the full operator space. Consequently, as the system dimension increases, the fraction of the total operator space spanned by local unitaries approaches one.

A broader class of probe operations is given by arbitrary deterministic maps, i.e.\ completely positive and trace-preserving (CPTP) maps. In the Choi representation, trace preservation requires that tracing over the output space yields the identity on the input space. Equivalently, in an expansion in a generalised Pauli operator basis, all coefficients associated with the terms of the form $\sigma_i\otimes\mathbb{I}$ with $i\neq 0$ must vanish. The space spanned by deterministic operations therefore has dimension only $d^2(d^2-1)+1$ (see the Supplementary Material~\cite{supplementary} for more details). Since this is still smaller than the full dimension $d^4$ required for informational completeness, deterministic operations alone are insufficient for complete reconstruction of a multi-time process. Therefore, an informationally complete probe set must \emph{necessarily include non-deterministic operations}. This is one of the main obstacles to experimental characterisation of multi-time processes. In practice, implementing such operations typically requires mid-circuit measurement and conditional re-preparation, which are substantially more demanding than gate-only interventions. As a result, most existing experiments have been limited to partial characterisation of multi-time processes \cite{White2020,White2022,White2025whatcanunitary,Milz_PRA_2018}. Full multi-time tomography has only recently been demonstrated experimentally on superconducting devices \cite{giarmatzi2023multitime}, although, the mid-circuit measurement and re-preparation steps contribute additional overhead and can introduce further errors through idling and measurement-induced crosstalk. It is therefore highly desirable to realise informationally complete probe sets using alternatives to mid-circuit measurement.

A standard example of a non-deterministic operation is a measure-and-prepare map, which in the Choi representation, is described by the operator $|a\rangle\langle a|^{T}\otimes |\psi\rangle\langle\psi|,$
where \( |a\rangle \) is a state vector specifying the rank-one measurement effect on the input space, \( |\psi\rangle \) is the prepared output state, and the transpose is taken with respect to the fixed computational basis. Such operations are sufficient to form an informationally complete set, since their linear span covers the full operator space. A straightforward coherent implementation of measure-and-prepare uses an ancilla of the same dimension as the system together with a SWAP unitary. Specifically, if the ancilla is prepared in the state $|\psi\rangle$, and after the SWAP interaction one measures the ancilla with POVM element $|a\rangle\langle a|$, then the induced operation on the system has Choi operator $|a\rangle\langle a|^{T}\otimes |\psi\rangle\langle\psi|$, exactly reproducing the desired measure-and-prepare map (see Fig.~\ref{fig:MP_all} (a,b)). Thus, access to an ancilla of dimension $d$ is sufficient to realise an informationally complete probe set for a single $d$-dimensional lab. This naturally raises the question: what is the smallest ancilla dimension required to realise such a set? The following theorem shows that a \emph{single qubit} ancilla already suffices.

\begin{theorem}\label{Th1}
An informationally complete set of operations for a single lab acting on a $d$-dimensional system can be implemented using a single qubit ancilla, coherent system--ancilla interactions, and measurement of the ancilla.
\end{theorem}

\begin{proof}
Consider a qubit ancilla prepared in the state $\ket{0}$, and let $U_{SA}$ be a joint unitary acting on system and ancilla. Measuring the ancilla in the computational basis $\{\ket{m}\}_{m=0,1}$ induces Kraus operators $K_{m,0}=(\mathbb{I}\otimes \bra{m})U_{SA}(\mathbb{I}\otimes \ket{0})$. The corresponding Choi operator is the rank-one projector $\ketbra{K_{m,0}}{K_{m,0}}$, where $\ket{K_{m,0}}=(\mathbb{I}\otimes K_{m,0})\ket{\mathbb{I}}$. Moreover, consider the following joint unitary  $U_{SA}=(V\otimes \mathbb{I})C_X(U\otimes \mathbb{I})$, where $U$ and $V$ are unitaries acting on the system, and $C_X=\ketbra{0}{0}\otimes \mathbb{I}+(\mathbb{I}-\ketbra{0}{0})\otimes X$, with $X$ the Pauli $X$ operator acting on the ancilla (see Fig.~\ref{fig:Single_lab_IC_setup}). For the measurement outcome $m=0$, one obtains $K_{0,0}=V\ketbra{0}{0}U$. With the following choice of local unitaries, i.e., $U^\dagger\ket{0}=\ket{a^*}$ and $V\ket{0}=\ket{\psi}$, we obtain $K_{0,0}=\ketbra{\psi}{a^*}$. The corresponding Choi operator is therefore $\ketbra{K_{0,0}}{K_{0,0}}=\ketbra{a}{a}^{T}\otimes \ketbra{\psi}{\psi}$, which is precisely the Choi form of a measure-and-prepare operation and such operators span the full operator space. Hence the accessible set is informationally complete.
\end{proof}

\begin{figure}
    \centering
    \includegraphics[width=0.9\linewidth]{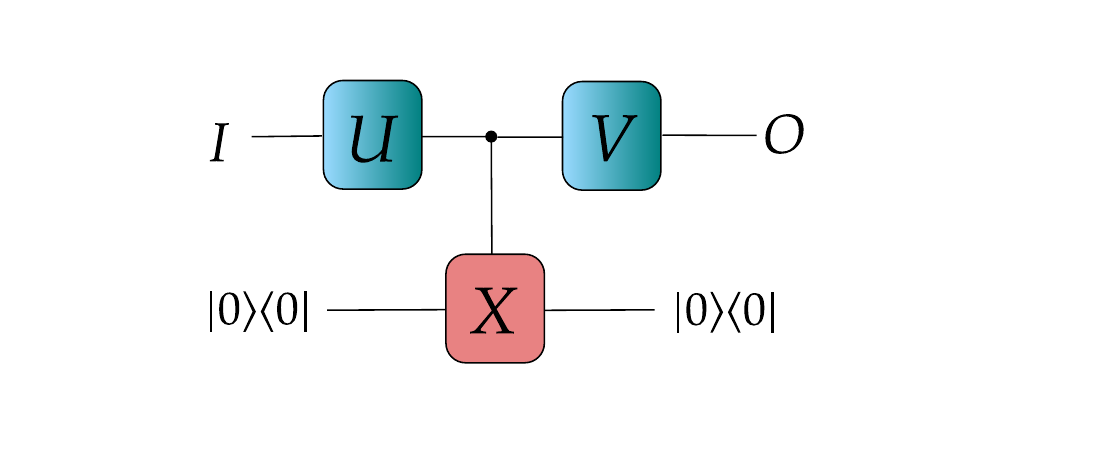}
    \caption{System--ancilla unitary construction for a single lab operation. The upper wire denotes the system and the lower wire a single-qubit ancilla. If $U^\dagger\ket{0}=\ket{a^*}$ and $V\ket{0}=\ket{\psi}$, then the ancilla measurement outcome $\ketbra{0}{0}$ induces the Kraus operator $K_{0,0}=\ketbra{\psi}{a^*}$, whose Choi operator is $\ketbra{K_{0,0}}{K_{0,0}}=\ketbra{a}{a}^{T}\otimes\ketbra{\psi}{\psi}$. Varying $\ket{a}$ and $\ket{\psi}$ generates a family of measure-and-prepare operations that spans the full operator space. }
    \label{fig:Single_lab_IC_setup}
\end{figure}

For the simplest experimental scenario with a qubit system i.e., both the input and output dimensions are two, an informationally complete set requires $16$ linearly independent operations. A convenient choice is to use ten unitary operations together with six non-deterministic operations. In particular, the set of unitaries ($I$, $X$, $Y$, $Z$, $R_X(\pi/2)$, $R_Y(\pi/2)$, $R_Z(\pi/2)$, $H$, $R_Z(\pi/2)X$, $R_X(\pi/2)Y$) results in linearly independent Choi operators and is readily implementable with high fidelity, here $R_{\sigma}(\theta)$ denotes a rotation by angle $\theta$ about the axis $\sigma$. The remaining six operations can be realised by preparing the ancilla in the positive eigenstate of a nontrivial Pauli operator, applying a SWAP unitary, and then measuring the ancilla in the same Pauli basis. Repeating this construction for the three Pauli bases yields six non-deterministic operations. Together, these sixteen operations form an informationally complete probe set for a single intermediate qubit lab.

\begin{figure}[t]
    \centering

    \begin{minipage}[t]{0.48\columnwidth}
        \centering
        \includegraphics[width=\linewidth]{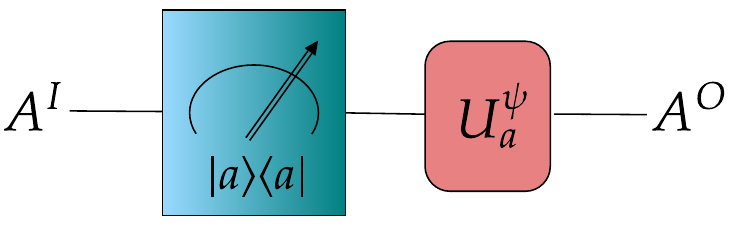}
        \textbf{(a)}
    \end{minipage}
    \hfill
    \begin{minipage}[t]{0.48\columnwidth}
        \centering
        \includegraphics[width=\linewidth]{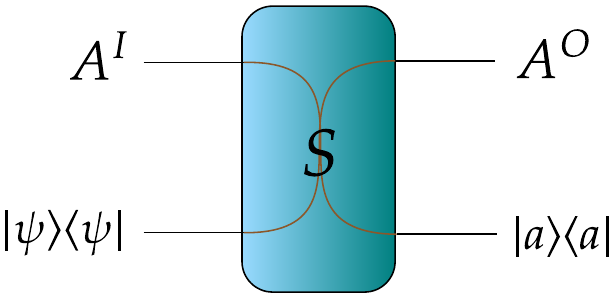}
        \textbf{(b)}
    \end{minipage}

    \vspace{0.25cm}

    \begin{minipage}[t]{0.98\columnwidth}
        \centering
        \includegraphics[width=\linewidth]{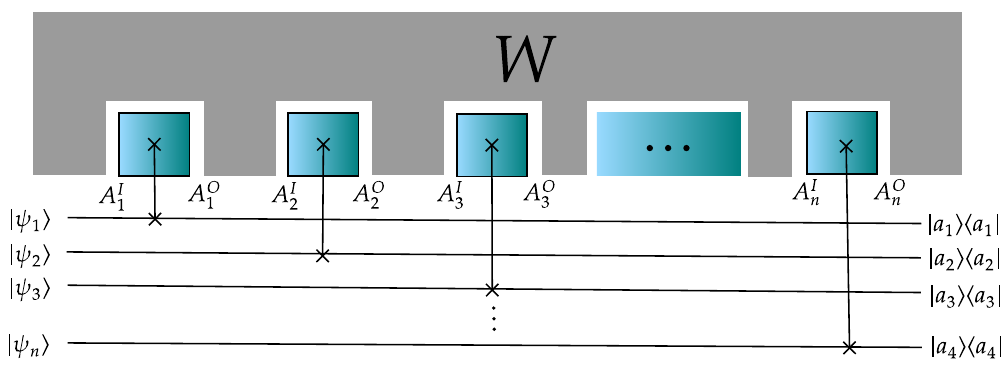}
        \textbf{(c)}
    \end{minipage}

    \caption{(a) Schematic depiction of measure and rotation resulting in the operations $|a\rangle\langle a|^T\otimes |\psi\rangle\langle\psi|$, which forms an informationally complete set of operations, however, difficult to implement in present devices. (b) Alternative implementation of measure and prepare operation through an ancilla and swap operation. (c) Measure and prepare in all the intermediate probes can be implemented if there are no restrictions on the number and the dimension of ancilla. However, it comes with a significant resource cost in addition to the implementation issue pertaining to applying SWAP with non-neighbouring ancillas.}
    \label{fig:MP_all}
\end{figure}

\noindent \textit{IC for multiple intermediate labs}: The preceding discussion shows that, if one has access to a $d$-dimensional ancilla at every intermediate lab, then measure-and-prepare operations can be implemented independently at each intervention, and informational completeness follows immediately  (see Fig.~\ref{fig:MP_all}(c)). By Theorem~\ref{Th1}, this requirement can be reduced further: a single qubit ancilla per lab already suffices to realise an informationally complete set of operations independently at each intervention. However, this still assumes access to a fresh ancilla for every lab. We now turn to the main question of this work --- what is the minimal ancillary resource required to implement informationally complete operations across an arbitrary number of labs? Related questions about minimal ancilla requirements for implementing general quantum dynamics have been studied previously, but typically under the assumption that the ancilla can be reset between uses \cite{Viola2001_PRA}. In our setting such resetting is not allowed, since we work without mid-circuit measurement. Indeed, allowing reset would effectively provide independent ancillas for each probe, in which case the implementation of an informationally complete set becomes straightforward.

Our setting is more restrictive than allowing independent operations at each lab. We consider a single ancilla prepared initially in a pure state, interacting unitarily with the system at successive probes, and measured only at the end. The resulting object is therefore not a product of independent lab operations, but a correlated multi-time probe, i.e.\ a superinstrument element, whose structure is constrained by causality \cite{taranto2025higherorderquantumoperations}. For example, for two intermediate labs $A_1$ and $A_2$ as shown in Fig.~\ref{fig:Testor}, preparing the ancilla in a pure state $|\psi\rangle$, letting it interact with the system through joint unitaries $U_1$ and $U_2$, and finally measuring it with outcome $a$, induces a correlated operation
\begin{equation}
    \label{eq:superinstrument}
    T_a = |\psi\rangle\langle\psi|\star [U_1]\star[U_2]\star|a\rangle\langle a|,
\end{equation}
where $[U]=|U\rangle\langle U|$ denotes the Choi operator of the unitary channel $U$, and $\star$ is the link product, i.e.\ contraction over the common ancilla spaces. This construction imposes a strong structural restriction on the accessible probes. Since all correlations between different labs are mediated solely by the ancilla, each $T_a$ admits a matrix-product-operator (MPO) representation whose bond dimension across any bipartition of the labs is at most $d_a^2$, where $d_a$ is the ancilla dimension. Product probes, corresponding to independent quantum operations at each lab with no ancilla-mediated memory or feedforward between probes, have bond dimension $1$, while genuinely correlated probes require nontrivial bond dimension. However, not every MPO of bond dimension at most $d_a^2$ is physically realisable by the sequential ancilla construction in Fig.~\ref{fig:Testor}, because the inter-lab correlations must arise from a single coherently evolving ancilla and are therefore more constrained than a generic low-bond-dimension MPO. The important question is therefore whether this  restricted family of correlated probes can nevertheless span the full multi-time operator space. We now turn to our main result which answers the above in the affirmative.

\begin{theorem}\label{Th2}
A single qubit ancilla suffices to implement an informationally complete set of probe operations for an arbitrary number of intermediate labs acting on $d$-dimensional systems, using only joint system--ancilla unitaries at the probes and a single final ancilla measurement.
\end{theorem}

\begin{proof}
We provide a constructive proof for two intermediate labs \(A_1\) and \(A_2\); the extension to an arbitrary number of labs is given in the Supplementary Material~\cite{supplementary}. Fix the ancilla input state and final measurement outcome to be \(\ket{0}\). If \(U_1\) and \(U_2\) are the joint system--ancilla unitaries at the two labs, then the resulting two-lab probe is
\begin{align}
T_{00}
=
\sum_{\alpha,\beta=0}^{1}
\ketbra{K^{(1)}_{\alpha,0}}{K^{(1)}_{\beta,0}}
\otimes
\ketbra{K^{(2)}_{0,\alpha}}{K^{(2)}_{0,\beta}},
\end{align}
where
$K^{(i)}_{m,n}
=
(I^{S}\otimes \bra m)\,U_i\,(I^{S}\otimes \ket n)$
are the \(d\times d\) ancilla blocks of \(U_i\) which takes the form (see Supplementary Material ~\cite{supplementary} for proof)
\begin{equation}
\label{Block_unitary}
U_{n}
=
\begin{pmatrix}
K_{00}^{(n)} & \sqrt{I-K_{00}^{(n)}K_{00}^{(n)\dagger}}\,V \\
W\sqrt{I-K_{00}^{(n)\dagger}K_{00}^{(n)}} & -W K_{00}^{(n)\dagger}V
\end{pmatrix},
\end{equation}
where \(V\) and \(W\) are arbitrary unitary operators on \(\mathcal{H}_d\). Therefore, \(T_{00}\) is a sum of product operators on the two lab spaces. In order to isolate a single product term, insert between the two interactions an ancilla phase gate
\begin{align}
R_{\theta}=|0\rangle\langle 0|+e^{i\theta}|1\rangle\langle 1|,
\end{align}
which updates \(K^{(1)}_{1,0}\) by  a phase term \(e^{i\theta}\), while leaving \(K^{(1)}_{0,0}\) unchanged, and therefore produces a family of probes \(T_{00}(\theta)\). The linear combination
\begin{align}
f_\theta(T_{00})
=
\frac14\Bigl(
T_{00}(0)-T_{00}(\pi)-iT_{00}(\pi/2)+iT_{00}(-\pi/2)
\Bigr)
\end{align}
filters out exactly the component with ancilla branch change \(0\to 1\) in the first lab and \(1\to 0\) in the second, yielding the following product term
\begin{align}
f_\theta(T_{00})
=
\ketbra{K^{(1)}_{1,0}}{K^{(1)}_{0,0}}
\otimes
\ketbra{K^{(2)}_{0,1}}{K^{(2)}_{0,0}}.
\end{align}
Consequently, it suffices to show that each factor can be chosen arbitrarily from a basis of the single-lab operator space \(\mathcal{L}(\mathcal H_d\otimes \mathcal H_d)\). To this end, we use the block form of a joint unitary with qubit ancilla given in Eq.~\eqref{Block_unitary}. In particular, exploiting the freedom that unitary $W$ and $V$ are arbitrary, one may choose the blocks so that
\begin{align}
K^{(1)}_{0,0}=\frac{\sigma_\nu}{\sqrt2},
\qquad
K^{(1)}_{1,0}=\frac{\sigma_\mu}{\sqrt2},
\end{align}
and similarly
\begin{align}
K^{(2)}_{0,0}=\frac{\sigma_{\nu'}}{\sqrt2},
\qquad
K^{(2)}_{0,1}= \frac{\sigma_{\mu'}}{\sqrt{2}},
\end{align}
where \(\{\sigma_\mu\}_{\mu=0}^{d^2-1}\) is a generalized Pauli (Heisenberg--Weyl) operator basis on \(\mathcal{L}(\mathcal H_d)\). It follows that $\ketbra{K^{(1)}_{1,0}}{K^{(1)}_{0,0}}
\propto
\ketbra{\sigma_\mu}{\sigma_\nu}$ and  $\ketbra{K^{(2)}_{0,1}}{K^{(2)}_{0,0}}
\propto
\ketbra{\sigma_{\mu'}}{\sigma_{\nu'}}.$
Since the operators $\left\{\ketbra{\sigma_\mu}{\sigma_\nu}\right\}_{\mu,\nu=0}^{d^2-1}$
form a basis of \(\mathcal{L}(\mathcal H_d\otimes \mathcal H_d)\), their tensor products span the full two-lab operator space. Therefore the family of probes generated by a single qubit ancilla is informationally complete for two intermediate labs. The proof for arbitrary many labs follows by iterating the same phase-filtering argument; see the Supplementary Material~\cite{supplementary}.
\end{proof}

\noindent \textit{Implication for quantum control:}
A central takeaway of Theorem~ \ref{Th2} is that the ancilla-assisted family of probe operations is not only sufficient for reconstructing the full multi-time process matrix, but also sufficient for evaluating \emph{any} linear functional of the process. This has direct consequences for control, verification, and memory quantification tasks, where the quantity of interest is typically expressed as the expectation value of a multi-time observable (tester) against the underlying process. Let \(W\) denote the unknown multi-time process and let \(O\) be any Hermitian operator on the same multi-time space, representing, for example, a tester effect, witness, or control objective. The associated linear functional is
$\langle O\rangle_W = \Tr[WO].$
Since the ancilla-assisted probes \(\{T_a\}\) generated by our protocol span the full operator space, any such \(O\) admits an expansion \(O=\sum_a c_a T_a\), and hence
$\Tr[WO]=\sum_a c_a\,\Tr[WT_a].$
Therefore, any linear multi-time figure of merit can be inferred directly from the same experimentally accessible statistics used for tomography, without directly implementing \(O\). Thus the control objective can be estimated, and in principle optimised, directly from accessible probe statistics, without full process reconstruction. This is especially useful when the desired figure of merit has a simple and well-conditioned expansion in the probe family. Viewed in this way, the protocol is a temporal analogue of circuit cutting and stitching \cite{Peng_PRL_2020,Tang_2021,Piveteau_2024}. Rather than directly implementing a complicated multi-time tester or control observable, one measures a spanning family of experimentally simple ancilla-assisted probes and reconstructs the desired functional by classical post-processing. The tradeoff is analogous in spirit, where hardware-demanding operations are replaced by a larger number of simpler experimental configurations together with additional classical processing.
 
\begin{figure}[t]
    \centering

   \begin{minipage}[t]{0.60\columnwidth}
        \centering
        \includegraphics[width=0.75\linewidth]{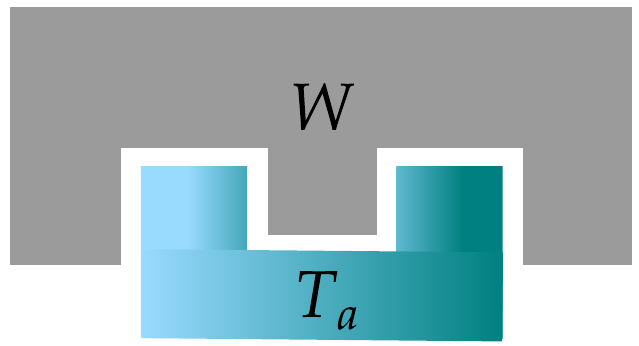}
        \textbf{(a)}
    \end{minipage}
    \vspace{0.25cm}

    \begin{minipage}[t]{0.75\columnwidth}
        \centering
        \includegraphics[width=\linewidth]{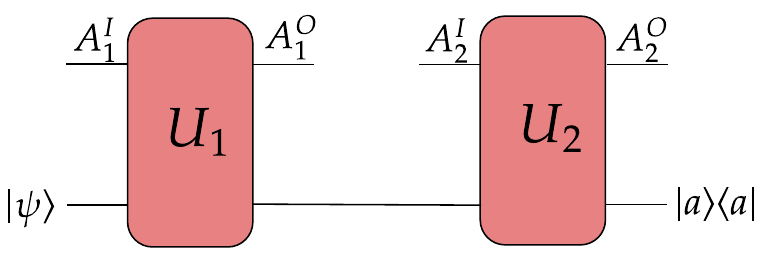}
        \textbf{(b)}
    \end{minipage}

    \caption{(a) A depiction of correlated operation across two intermediate labs $T_a$ to characterise the process $W$ (b) An implementation of the correlated instrument $T_a$ through continuous interaction with a single ancilla. The resulting operation has a restricted MPO structure.}
    \label{fig:Testor}
\end{figure}

\noindent \textit{Conclusion:}  We have shown that a single coherent qubit ancilla is sufficient to implement an informationally complete family of probes for multi-time processes of arbitrary length and system dimension, without the need for mid-circuit measurement, feed-forward, or reset. This identifies a minimal and experimentally relevant route to complete multi-time process tomography on current architectures. Since the resulting probe family spans the full multi-time operator space, the same construction also allows arbitrary linear functionals of a process, such as witnesses and control objectives, to be inferred from accessible probe statistics. It is important to note that while informational completeness is a geometric statement about span, the practical usefulness of a probe family in finite-shot experiments also depends on its conditioning and noise robustness. It would therefore be important to compare the statistical performance of different experimentally accessible probe families, including the ancilla-assisted construction introduced here. Although unitary probes are insufficient for complete tomography at any fixed finite dimension, they span an increasingly large fraction of the full operator space in high dimensions. This raises a natural question: which physically relevant quantities can still be estimated efficiently from such restricted data? In particular, even without full informational completeness, it may be possible to access important process parameters, witnesses, or control objectives that lie largely within the unitary-accessible subspace. Clarifying the identifiability of process features under restricted probe families is therefore an important direction for future work. Some aspects of this question, in the context of memory detection, were investigated in Ref.~\cite{White2025whatcanunitary}.

\acknowledgments

We thank Kamudibikash Goswami, Gavin Brennen, Fabio Costa and Kavan Modi for interesting discussions. AKR and VS acknowledge funding support from Sydney Quantum Academy. This work was partially supported by the Hon-Hai Research Institute through the Australian Quantum Software Network (AQSN) micro-grant. 

\bibliography{Non_Markovianity}
\clearpage

\setcounter{page}{1}
\renewcommand{\thepage}{S\arabic{page}}

\setcounter{equation}{0}
\renewcommand{\theequation}{S.\arabic{equation}}

\setcounter{figure}{0}
\renewcommand{\thefigure}{S\arabic{figure}}

\setcounter{table}{0}
\renewcommand{\thetable}{S\arabic{table}}

\section*{Supplemental Material}
\addcontentsline{toc}{section}{Supplemental Material}

\input{supplementary}

\end{document}

%% file: preamble.tex
\usepackage{algorithm}
\usepackage{algpseudocode}
\usepackage[T1]{fontenc}
\usepackage{inputenc}
\usepackage{xcolor}
\usepackage{babel}
\usepackage{dsfont}
\usepackage{physics}
\usepackage{mathtools}
\usepackage{bbding}
\usepackage{pifont}
\usepackage{textcomp}
\usepackage{wasysym}
\usepackage{amsthm}
\usepackage[normalem]{ulem}
\usepackage{nicefrac}
\usepackage{wasysym}
\usepackage{dsfont}
\usepackage{amsmath}
\usepackage{amssymb}
\usepackage{graphicx}

\usepackage{soul}
\definecolor{cream}{rgb}{1.0, 0.99, 0.82}
\definecolor{seafoam}{rgb}{	0.576, 0.914, 0.745}
\definecolor{olive}{rgb}{0.522, 0.666, 0.526}

\usepackage{hyperref}

\hypersetup{
     colorlinks   = true,
     linkcolor    = blue,
     citecolor    = blue,
     urlcolor     = blue,
     }

\usepackage{appendix}
\renewcommand{\selectlanguage}[1]{}
\makeatletter
\@ifundefined{textcolor}{}
{%
	\definecolor{BLACK}{gray}{0}
	\definecolor{WHITE}{gray}{1}
	\definecolor{RED}{rgb}{1,0,0}
	\definecolor{GREEN}{rgb}{0,1,0}
	\definecolor{BLUE}{rgb}{0,0,1}
	\definecolor{CYAN}{cmyk}{1,0,0,0}
	\definecolor{MAGENTA}{cmyk}{0,1,0,0}
	\definecolor{YELLOW}{cmyk}{0,0,1,0}
}
\theoremstyle{plain}

\theoremstyle{plain}

\ifx\proof\undefined
\newenvironment{proof}[1][\protect\proofname]{\par
	\normalfont\topsep6\p@\@plus6\p@\relax
	\Trivlist
	\itemindent\parindent
	\item[\hskip\labelsep
	\scshape
	#1]\ignorespaces
}{%
	\endtrivlist\@endpefalse
}
\providecommand{\proofname}{\textbf{Proof}}
\fi
\theoremstyle{plain}

\renewenvironment{proof}[1][\proofname]{\noindent {\bfseries #1.} }{\qed}

\providecommand{\lemmaname}{Lemma}
\providecommand{\definitionname}{Definition}
\providecommand{\propositionname}{Proposition}

\usepackage{babel}
\usepackage{txfonts}
\usepackage{colortbl}\definecolor{myurlcolor}{rgb}{0,0,0.7}

\renewcommand{\bra}[1]{\left\langle #1 \right|}
\renewcommand{\ket}[1]{\left| #1 \right\rangle}
\renewcommand{\braket}[2]{\left\langle #1 \middle| #2 \right\rangle}
\renewcommand{\ketbra}[2]{\left|#1\middle\rangle\!\middle\langle#2\right|}
\newcommand{\proj}[1]{\ketbra{#1}{#1}}

\newcommand{\haH}

\newtheorem{theorem}{Theorem}
\newtheorem*{thm*}{Theorem}

\newtheorem{lemma}{Lemma}
\newtheorem*{lem*}{Lemma}
\theoremstyle{definition}

\newtheorem*{cor*}{Corollary}

\newtheorem*{prop*}{Proposition}

\makeatother

%% file: supplementary.tex
The notations and symbols used here are the same as those in the main manuscript and carry the same meanings.

\section{Preliminaries}

In this section, we briefly review some preliminary material.

\subsection{Link product.}
Consider \(A\in\mathcal L(\mathcal H_A)\) and \(B\in\mathcal L(\mathcal H_B)\), where
\[
\mathcal H_A=\mathcal H_{A\setminus B}\otimes \mathcal H_{A\cap B},
\qquad
\mathcal H_B=\mathcal H_{A\cap B}\otimes \mathcal H_{B\setminus A}
\]
i.e., the operators \(A\) and \(B\) share the subsystem \(\mathcal H_{A\cap B}\). The \emph{link product} of \(A\) and \(B\), denoted \(A\star B\), is defined as \cite{chiribella09b,taranto2025higherorderquantumoperations}
\begin{equation}\label{Eq:Link_product}
A\star B
:=
\Tr_{A\cap B}
\Big[
\big(A^{T_{A\cap B}}\otimes I_{B\setminus A}\big)
\big( I_{A\setminus B}\otimes B\big)
\Big].
\end{equation}
Here, \(T_{A\cap B}\) denotes the partial transpose on the common subsystem \(\mathcal H_{A\cap B}\), and \(\Tr_{A\cap B}\) denotes the partial trace over the same subsystem. The resulting operator \(A\star B\) acts on
\[
\mathcal H_{A\setminus B}\otimes \mathcal H_{B\setminus A}.
\]
Operationally, the link product provides the natural composition rule for Choi operators ,i.e, it contracts two operators over their common input-output space and yields the Choi operator of the composed map. The link product satisfies several useful properties. First, if \(A\) and \(B\) are Hermitian, then \(A\star B\) is also Hermitian. Second, if \(A\ge 0\) and \(B\ge 0\), then \(A\star B\ge 0\). Thus, the link product preserves both Hermiticity and positivity, which is essential when composing valid quantum operations. Moreover, the link product is associative whenever the composition is well-defined, namely
\[
A\star(B\star C)=(A\star B)\star C,
\]
provided $\mathcal{H}_{A} \cap \mathcal{H}_{B} \cap \mathcal{H}_{C}=\phi$. Finally, it is commutative up to the canonical identification of tensor factors:
\[
A\star B \equiv B\star A,
\]
that is, the two expressions coincide up to a relabelling (or swap) of the remaining Hilbert spaces. 

\subsection{Constructing process matrices and superinstruments via link product}

The link product provides a compact way to write both process matrices and superinstruments directly in Choi form. Although the two constructions are analogous, they represent different higher-order operations.

\paragraph{Process matrices.}
Consider a multi-time process generated by an initial environment state \(\rho^{E_0}\) and a sequence of joint system-environment unitaries
\[
U^{(0)}_{SE},U^{(1)}_{SE},\dots,U^{(N)}_{SE}.
\]
The associated process matrix \(W\) is obtained by linking the initial environment state with all the joint unitaries and finally tracing out the last environment degree of freedom as the following \cite{taranto2025higherorderquantumoperations,Goswami_2025}
\begin{equation}
W=
\rho^{E_0}\star [U^{(0)}_{SE}]\star [U^{(1)}_{SE}]\star \cdots \star [U^{(N)}_{SE}]\star I^{E_N}.
\label{eq:process_link}
\end{equation}
Here, the final factor \(I^{E_N}\) implements the trace over the final environment system. Moreover, $[U^{(t)}_{SE}]\in \mathcal{L}(\mathcal{H}_{t}^{O}\otimes \mathcal{H}_{t+1}^{I})$
where the superscripts \(O\) and \(I\) denote, respectively, the output and input Hilbert spaces at time \(t\) and $t+1$ respectively. Therefore, the process matrix belongs to $W\in \mathcal{L}\left(\bigotimes_{t=0}^{N}\mathcal{H}_{t}^{O}\otimes \mathcal{H}_{t+1}^{I}\right)$.

By construction, \(W\geq 0\). In addition, \(W\) satisfies the linear causality constraints of a quantum comb. To state these, define the reduced operators \(W_t\), for \(t=1,\dots,N\), which belongs to the space
$\mathcal{L}\left(\bigotimes_{j=0}^{t}\mathcal{H}_{j}^{O}\otimes \mathcal{H}_{j+1}^{I}\right)$
with \(W_N=W\). The causality constraints then take the recursive form \cite{chiribella08,pollock_operational_markov,Pollock_pra,taranto2025higherorderquantumoperations}
\begin{equation}
\operatorname{Tr}_{\mathcal{H}_{t+1}^{I}}[W_t]
=
I \otimes W_{t-1},
\qquad t=1,\dots,N
\label{eq:comb_constraint}
\end{equation}
Equivalently, a valid process matrix is precisely a positive operator on
$\mathcal{L}(\bigotimes_{t=0}^{N}\mathcal{H}_{t}^{O}\otimes \mathcal{H}_{t+1}^{I})$
satisfying the hierarchy of linear trace conditions in Eq.~\eqref{eq:comb_constraint}.

\paragraph{Superinstruments/testers}
Superinstruments (or testers) are the multi-time generalisation of POVMs and quantum instruments. Just as an ordinary instrument is a collection of completely positive maps that sum to a deterministic channel, a superinstrument is a collection of positive higher-order maps that sum to a deterministic higher-order transformation. In particular, by allowing some of the Hilbert spaces to be trivial, one recovers ordinary instruments and POVMs as special cases. 

In the present setting, a superinstrument is constructed similarly to a process matrix, except that the inaccessible environment is replaced by a controllable ancilla. Let the ancilla be prepared in the state \(\ket{\psi}\), interact sequentially with the system through joint system-ancilla unitaries
$U^{(1)}_{SA},U^{(2)}_{SA},\dots,U^{(N)}_{SA},$
and finally be measured projectively with outcome \(m\). The corresponding superinstrument element is

\begin{equation}
\mathcal{T}_{\psi,m}
=
\ketbra{\psi}{\psi}\star [U^{(1)}_{SA}]\star [U^{(2)}_{SA}]\star \cdots \star [U^{(N)}_{SA}]\star \ketbra{m}{m},
\label{eq:superinstrument_link}
\end{equation}
where $[U^{(t)}_{SA}]\in \mathcal{L}(\mathcal{H}_{t}^{I}\otimes \mathcal{H}_{t}^{O})$. Unlike Eq.~\eqref{eq:process_link}, where the final environment is traced out, here the final ancilla degree of freedom is measured, producing a family \(\{\mathcal{T}_{\psi,m}\}_m\) of probabilistic higher-order maps. This is precisely the higher-order analogue of realising an ordinary instrument by coupling to an ancilla and measuring it at the end.  Each element \(\mathcal{T}_{\psi,m}\) is positive semi-definite $\mathcal{T}_{\psi,m}\ge 0,$, and the sum over outcomes yields a deterministic supermap,
\begin{equation}
\sum_m \mathcal{T}_{\psi,m}
=
\ketbra{\psi}{\psi}\star [U^{(1)}]\star [U^{(2)}]\star \cdots \star [U^{(N)}]\star I^{A_N},
\label{eq:det_superinstrument_link}
\end{equation}
where \(I^{A_N}\) implements the trace over the final ancilla. However, unlike the process matrix case, the deterministic object \(\sum_m \mathcal{T}_{\psi,m}\) satisfies the comb constraints in ~\eqref{eq:comb_constraint} with the input and output labels interchanged.  Finally, when a superinstrument is contracted with a valid process matrix, the result is a valid probability distribution over outcomes. In this sense, superinstruments are exactly the higher-order objects that extract classical information from multi-time processes, just as POVMs and instruments do in the single-time setting. 

\section{Span of local unitaries}

As mentioned in the main text, for a \(d\)-dimensional system, the linear span of Choi operators corresponding to unitary channels has dimension $(d^2-1)^2+1.$
In this section, we provide a proof of this statement. Consider \(\{\sigma_\alpha\}_{\alpha=0}^{d^2-1}\) be a Hilbert--Schmidt orthonormal operator basis on \(\mathcal H_d\), normalized as
\begin{equation}
\Tr[\sigma_\alpha^\dagger \sigma_\beta]=\delta_{\alpha\beta},
\label{eq:HS_orthonormal}
\end{equation}
where \(\sigma_0=I/\sqrt d\), and \(\sigma_i\) for \(i=1,\dots,d^2-1\) are traceless. Throughout, Roman indices \(i,j,p,q\) run only over the traceless elements \(1,\dots,d^2-1\). For a unitary \(U\), let \([U]\) denote the Choi operator of the corresponding unitary channel. Expanding \([U]\) in the product basis \(\{\sigma_\alpha^T\otimes \sigma_\beta\}\), we may write
\begin{equation}
[U]=\sum_{\alpha,\beta=0}^{d^2-1} C_{\alpha\beta}(U)\,\sigma_\alpha^T\otimes \sigma_\beta ,
\label{eq:choi_expansion_general}
\end{equation}
for suitable coefficients \(C_{\alpha\beta}(U)\). For a unitary channel, the Choi operator satisfies
\begin{equation}
\Tr_O [U]=I,
\qquad
\Tr_I [U]=I.
\label{eq:unitary_partial_trace_constraints}
\end{equation}
Using the expansion in Eq.~\eqref{eq:choi_expansion_general}, together with the fact that $\Tr[\sigma_0]=\sqrt d$ and
$\Tr[\sigma_i]=0 \quad \forall i\ge 1,$
we find that the conditions in Eq.~\eqref{eq:unitary_partial_trace_constraints} force
\begin{equation}
C_{0i}(U)=0,
\qquad
C_{i0}(U)=0
\qquad \forall i=1,\dots,d^2-1.
\label{eq:C0i_Ci0_zero}
\end{equation}
Therefore, the Choi operator of any unitary has the form
\begin{equation}
[U]
=
\frac{1}{d}\,I\otimes I
+
\sum_{i,j=1}^{d^2-1}
C_{ij}(U)\,\sigma_i^T\otimes \sigma_j .
\label{eq:unitary_choi_restricted}
\end{equation}
Equation~\eqref{eq:unitary_choi_restricted} shows that no term of the form $\sigma_i^T\otimes I$ or
 $I\otimes \sigma_j$
can appear in the Choi operator of a unitary. Since these operators span a subspace of dimension \(2(d^2-1)\), it follows that the span of unitary Choi operators is orthogonal to this subspace. Hence,
\begin{equation}
\dim \mathrm{span}\{[U]: U \text{ unitary}\}
\le d^4-2(d^2-1)
=
(d^2-1)^2+1.
\label{eq:upper_bound_span}
\end{equation}
Therefore, \((d^2-1)^2+1\) is an upper bound. It remains to show that this bound is attained. To prove that the upper bound is saturated, it is enough to show that every basis element appearing in Eq.~\eqref{eq:unitary_choi_restricted}, namely $ I\otimes I/d
\quad \text{and} \quad
\sigma_q^T\otimes \sigma_p,$
lies in the linear span of Choi operators of unitaries. The identity term is immediate since every \([U]\) contains the component \(I\otimes I/d\), this direction is already present in the span.

Now let \(\{U_\ell\}_{\ell=1}^L\) be any unitary \(2\)-design, and define
\begin{equation}
K_{pq}
:=
\frac{1}{L}\sum_{\ell=1}^L
C_{pq}(U_\ell)\,[U_\ell],
\quad
p,q\in\{1,\dots,d^2-1\},
\label{eq:Kpq_definition}
\end{equation}
where
\begin{equation}
C_{pq}(U_\ell)
=
\Tr\!\big[\sigma_p\, U_\ell \sigma_q U_\ell^\dagger\big].
\label{eq:Cpq_definition}
\end{equation}
By construction, \(K_{pq}\) is a linear combination of Choi operators of unitaries, so belongs to their span. Substituting Eq.~\eqref{eq:unitary_choi_restricted} into Eq.~\eqref{eq:Kpq_definition} gives
\begin{equation}
K_{pq}
=
\frac{1}{L}\sum_{\ell=1}^L
C_{pq}(U_\ell)
\left(
\frac{1}{d}I\otimes I
+
\sum_{i,j=1}^{d^2-1}
C_{ij}(U_\ell)\,\sigma_i^T\otimes \sigma_j
\right).
\label{eq:Kpq_expand1}
\end{equation}
Since \(p,q\ge 1\), the coefficient \(C_{pq}(U_\ell)\) averages to zero over a unitary \(1\)-design, and the identity contribution vanishes. Therefore,
\begin{equation}
K_{pq}
=
\sum_{i,j=1}^{d^2-1}
\left(
\frac{1}{L}\sum_{\ell=1}^L
C_{pq}(U_\ell)\,C_{ij}(U_\ell)
\right)
\sigma_i^T\otimes \sigma_j .
\label{eq:Kpq_expand2}
\end{equation}

It remains to evaluate the coefficient
\begin{equation}
M_{pq,ij}
:=
\frac{1}{L}\sum_{\ell=1}^L
C_{pq}(U_\ell)\,C_{ij}(U_\ell).
\label{eq:Mpqij_definition}
\end{equation}
Using Eq.~\eqref{eq:Cpq_definition}, we may rewrite this as
\begin{equation}
M_{pq,ij}
=
\frac{1}{L}\sum_{\ell=1}^L
\Tr[\sigma_p U_\ell \sigma_q U_\ell^\dagger]\,
\Tr[\sigma_i U_\ell \sigma_j U_\ell^\dagger].
\label{eq:Mpqij_traces}
\end{equation}
Equivalently,
\begin{equation}
M_{pq,ij}
=
\Tr\!\left[
(\sigma_p\otimes \sigma_i)
\left(
\frac{1}{L}\sum_{\ell=1}^L
U_\ell^{\otimes 2}
(\sigma_q\otimes \sigma_j)
(U_\ell^\dagger)^{\otimes 2}
\right)
\right].
\label{eq:Mpqij_twirl}
\end{equation}

Since \(\{U_\ell\}\) is a unitary \(2\)-design, the average in Eq.~\eqref{eq:Mpqij_twirl} equals the Haar twirl of the operator \(\sigma_q\otimes \sigma_j\). The second-moment twirl has the form
\begin{equation}
\Phi_2(X)
=
\int dU\, U^{\otimes 2} X (U^\dagger)^{\otimes 2}
=
a(X)\,I + b(X)\,F,
\label{eq:haar_second_moment}
\end{equation}
where \(F\) is the swap operator. Because \(\sigma_q\) and \(\sigma_j\) are traceless, the \(I\)-component does not contribute, while orthonormality of the basis implies that the \(F\)-component selects exactly the matching indices. Consequently,
\begin{equation}
M_{pq,ij}
=
\delta_{pi}\delta_{qj}.
\label{eq:Mpqij_delta}
\end{equation}
Substituting Eq.~\eqref{eq:Mpqij_delta} into Eq.~\eqref{eq:Kpq_expand2}, we obtain
\begin{equation}
K_{pq}
=
\sigma_p^T\otimes \sigma_q .
\label{eq:Kpq_final}
\end{equation}
Therefore, every tensor-product basis element \(\sigma_p^T\otimes \sigma_q\) with \(p,q\ge 1\) lies in the span of unitary Choi operators. Together with the identity direction \(I\otimes I\), this shows that the span contains exactly $1+(d^2-1)^2$
linearly independent directions. We therefore conclude that
\begin{equation}
\dim \mathrm{span}\{[U]: U \text{ unitary}\}
=
(d^2-1)^2+1.
\label{eq:final_span_dimension}
\end{equation}
This proves the claim.

\section{IC set for arbitrary labs}

In this section, we provide a detailed proof of Theorem 2 from the main text.

\textbf{Theorem 2} \emph{A single qubit ancilla is sufficient to implement an informationally complete set of operations for an arbitrary number of probing laboratories, each acting on a \(d\)-dimensional input-output system, using only joint system--ancilla unitaries at the laboratories and a single final projective measurement on the ancilla.}

We consider \(N\) probing laboratories interacting sequentially with a common single-qubit ancilla. The corresponding superinstrument is constructed via the link product of the ancilla initial state, the Choi operators of the joint system--ancilla unitaries, and a final ancilla measurement. Explicitly,
\begin{equation}
\mathcal{T}_{\psi,m}
=
\proj{\psi}
\star
[U^{(1)}]
\star
[U^{(2)}]
\star \cdots \star
[U^{(N)}]
\star
\proj{m},
\label{eq:arb_labs_superinstrument}
\end{equation}
where \(\proj{\psi}\) is the initial ancilla state, \(\proj{m}\) labels the final projective measurement outcome, and $[U^{(n)}]$ is the Choi operator of the joint unitary \(U^{(n)}\) acting on the system and ancilla at laboratory \(n\).

We first use the following lemma.

\begin{lemma}
Each joint unitary \(U^{(n)}\) can be written, with respect to the computational basis \(\{\ket{0},\ket{1}\}\) of the ancilla, in the block form
\begin{equation}
U^{(n)}
=
\begin{pmatrix}
K_{00}^{(n)} & \sqrt{I-K_{00}^{(n)}K_{00}^{(n)\dagger}}\,V \\
W\sqrt{I-K_{00}^{(n)\dagger}K_{00}^{(n)}} & -W K_{00}^{(n)\dagger}V
\end{pmatrix},
\label{eq:block_unitary}
\end{equation}
where
\begin{equation}
K_{00}^{(n)} := (I\otimes\bra{0})U^{(n)}(I\otimes\ket{0}) \in \mathcal{L}(\mathcal{H}_d),
\end{equation}
and \(V\) and \(W\) are unitary operators on \(\mathcal{H}_d\).
\end{lemma}

\begin{proof}
Consider
\begin{equation}
K_{mn}^{(n)} = (I\otimes\bra{m})U^{(n)}(I\otimes\ket{n}),
\qquad m,n\in\{0,1\}.
\end{equation}
Since \(U^{(n)}\) is unitary, its block operators satisfy
\begin{align}
K_{00}^{(n)\dagger}K_{00}^{(n)} + K_{10}^{(n)\dagger}K_{10}^{(n)} &= I, \nonumber\\
K_{00}^{(n)}K_{00}^{(n)\dagger} + K_{01}^{(n)}K_{01}^{(n)\dagger} &= I, \label{eq:block_constraints}\\
K_{00}^{(n)}K_{10}^{(n)\dagger} + K_{01}^{(n)}K_{11}^{(n)\dagger} &= 0. \nonumber
\end{align}
From the first two relations, one may write
\begin{equation}
K_{01}^{(n)} = D^{(n)}V,
\qquad
K_{10}^{(n)} = W\overline{D}^{(n)},
\label{eq:offdiag_blocks}
\end{equation}
where
\begin{equation}
D^{(n)}=\sqrt{I-K_{00}^{(n)}K_{00}^{(n)\dagger}},
\qquad
\overline{D}^{(n)}=\sqrt{I-K_{00}^{(n)\dagger}K_{00}^{(n)}},
\end{equation}
and \(V,W\) are unitaries on \(\mathcal{H}_d\). Using the singular-value decomposition of \(K_{00}^{(n)}\), one verifies that
\begin{equation}
D^{(n)}K_{00}^{(n)\dagger}=K_{00}^{(n)\dagger}\overline{D}^{(n)}.
\label{eq:D_commutation}
\end{equation}
Substituting Eq.~\eqref{eq:offdiag_blocks} into the third relation of Eq.~\eqref{eq:block_constraints} then gives
\begin{equation}
K_{11}^{(n)}V^\dagger \overline{D}^{(n)} + W K_{00}^{(n)\dagger}\overline{D}^{(n)}=0,
\end{equation}
and hence
\begin{equation}
K_{11}^{(n)}=-W K_{00}^{(n)\dagger}V.
\label{eq:K11_form}
\end{equation}
This proves Eq.~\eqref{eq:block_unitary}.
\end{proof}

Without loss of generality, we choose the ancilla initial state and final measurement outcome to be \(\ket{\psi}=\ket{m}=\ket{0}\) since we can absorb the state changes into the corresponding unitaries. Then Eq.~\eqref{eq:arb_labs_superinstrument} becomes
\begin{equation}
\mathcal{T}_{0,0}
=
\sum_{\vec{\alpha},\vec{\beta}}
\ketbra{K_{\alpha_1 0}^{(1)}}{K_{\beta_1 0}^{(1)}}
\otimes
\ketbra{K_{\alpha_2 \alpha_1}^{(2)}}{K_{\beta_2 \beta_1}^{(2)}}
\otimes \cdots \otimes
\ketbra{K_{0\alpha_{N-1}}^{(N)}}{K_{0\beta_{N-1}}^{(N)}},
\label{eq:T00_expanded}
\end{equation}
where \(\vec{\alpha}=(\alpha_1,\dots,\alpha_{N-1})\) and \(\vec{\beta}=(\beta_1,\dots,\beta_{N-1})\), with each component taking values in \(\{0,1\}\). Next, for each \(n\in\{1,\dots,N-1\}\), define a modified unitary
\begin{equation}
\overline{U}^{(n)}=(I\otimes R_{\theta_n})U^{(n)},
\qquad
R_{\theta_n}=\proj{0}+e^{i\theta_n}\proj{1},
\label{eq:phase_modified_unitary}
\end{equation}
which induces a family of superinstruments \(\mathcal{T}_{0,0}(\vec{\theta})\), with \(\vec{\theta}=(\theta_1,\dots,\theta_{N-1})\). For each \(n\), define

\begin{align}
f_{\theta_n}\bigl(\mathcal{T}_{0,0}(\vec{\theta})\bigr)
&=
\frac{1}{4}\Big(
\mathcal{T}_{0,0}(\theta_n=0)
-
\mathcal{T}_{0,0}(\theta_n=\pi)
\nonumber\\
&\hspace{1cm}
-
i\,\mathcal{T}_{0,0}(\theta_n=\pi/2)
+
i\,\mathcal{T}_{0,0}(\theta_n=-\pi/2)
\Big).
\label{eq:phase_filter}
\end{align}
Applying these maps successively isolates the term
\begin{align}
&f_{\theta_{N-1}}\Bigl(
    f_{\theta_{N-2}}\bigl(
        \cdots f_{\theta_1}(\mathcal{T}_{0,0}(\vec{\theta}))
    \cdots
    \bigr)
\Bigr)
\nonumber\\
&\qquad =
\ketbra{K_{10}^{(1)}}{K_{00}^{(1)}}
\otimes
\ketbra{K_{11}^{(2)}}{K_{00}^{(2)}}
\otimes \cdots
\nonumber\\
&\qquad\quad \otimes
\ketbra{K_{11}^{(N-1)}}{K_{00}^{(N-1)}}
\otimes
\ketbra{K_{01}^{(N)}}{K_{00}^{(N)}}.
\label{eq:isolate_tensor_term}
\end{align}
Thus, it suffices to show that each factor in Eq.~\eqref{eq:isolate_tensor_term} can span an arbitrary basis element of \(\mathcal{L}(\mathcal{H}_d\otimes\mathcal{H}_d)\). To this end, let \(\{\sigma_\mu\}_{\mu=0}^{d^2-1}\) be a generalized Pauli (Weyl--Heisenberg) operator basis for \(\mathcal{L}(\mathcal{H}_d)\), normalized so that
\begin{equation}
\Tr[\sigma_\mu^\dagger \sigma_\nu]=d\,\delta_{\mu\nu}.
\label{eq:weyl_orthogonality}
\end{equation}
The vectorized operators \(\left\{\ket{\sigma_\mu}\right\}\) therefore form an orthogonal basis of \(\mathcal{H}_d\otimes\mathcal{H}_d\), since
\begin{equation}
\braket{\sigma_\mu}{\sigma_\nu}
=
\Tr[\sigma_\mu^\dagger \sigma_\nu].
\label{eq:vec_basis_identity}
\end{equation}
Hence the operators $\ketbra{\sigma_\mu}{\sigma_\nu}$
form a basis of \(\mathcal{L}(\mathcal{H}_d\otimes\mathcal{H}_d)\). For the first laboratory, choose
\begin{equation}
K_{00}^{(1)}=\frac{\sigma_\nu}{\sqrt{2}},
\qquad
W=\sigma_\mu .
\label{eq:first_lab_choice}
\end{equation}
Then
\begin{equation}
\overline{D}^{(1)}=\sqrt{I-K_{00}^{(1)\dagger}K_{00}^{(1)}}=\frac{I}{\sqrt{2}},
\end{equation}
so that
\begin{equation}
K_{10}^{(1)}=W\overline{D}^{(1)}=\frac{\sigma_\mu}{\sqrt{2}}.
\label{eq:first_lab_K10}
\end{equation}
Therefore,
\begin{equation}
\ketbra{K_{10}^{(1)}}{K_{00}^{(1)}}
=
\frac{1}{2}\ketbra{\sigma_\mu}{\sigma_\nu},
\label{eq:first_lab_basis}
\end{equation}
and by varying \(\mu,\nu\in\{0,\dots,d^2-1\}\), one obtains all basis elements. A similar argument applies to the final laboratory, where the relevant factor is $\ketbra{K_{01}^{(N)}}{K_{00}^{(N)}}$.
For each intermediate laboratory \(n\in\{2,\dots,N-1\}\), choose
\begin{equation}
K_{00}^{(n)}=\frac{\sigma_\nu}{\sqrt{2}},
\qquad
W=\sigma_\mu,
\qquad
V=I.
\label{eq:middle_lab_choice}
\end{equation}
Using Eq.~\eqref{eq:K11_form}, and the unitarity of \(\sigma_\nu\), we get
\begin{equation}
K_{11}^{(n)}=-W K_{00}^{(n)\dagger}V
=
-\frac{\sigma_\mu \sigma_\nu^\dagger}{\sqrt{2}}.
\label{eq:middle_lab_K11}
\end{equation}
Since products of Weyl operators are again Weyl operators up to a phase, one has
\begin{equation}
\sigma_\mu \sigma_\nu^\dagger = e^{i\phi_{\mu\nu}}\sigma_\lambda
\end{equation}
for some \(\lambda\). Hence
\begin{equation}
\ketbra{K_{11}^{(n)}}{K_{00}^{(n)}}
\propto
\ketbra{\sigma_\lambda}{\sigma_\nu},
\label{eq:middle_lab_basis}
\end{equation}
and again all basis elements are obtained by varying the Weyl labels. Therefore, each tensor factor in Eq.~\eqref{eq:isolate_tensor_term} can be chosen arbitrarily from a basis of \(\mathcal{L}(\mathcal{H}_d\otimes\mathcal{H}_d)\). Their tensor products therefore span the full operator space of the \(N\)-laboratory process. Consequently, the linear span of the obtainable superinstrument elements is informationally complete. This proves the theorem.
